\documentclass[11pt]{article}
\usepackage[a4paper, total={6in, 8in}]{geometry}
\RequirePackage{amsthm,amsmath,amsfonts,amssymb,hyperref}
\RequirePackage[numbers]{natbib}
\RequirePackage{graphicx}
\usepackage{multirow}
\usepackage{xcolor}
\usepackage{subcaption}
\usepackage{bbm}
\usepackage{natbib}

\newcommand{\ind}{{\rm I\hspace{-2.3mm} 1}}

\theoremstyle{plain}
\ifx\proposition\undefined
\newtheorem{proposition}{Proposition}
\fi

\ifx\lemma\undefined
\newtheorem{lemma}{Lemma}
\fi

\ifx\condition\undefined
\newtheorem{condition}{Condition}
\fi

\usepackage{authblk}

\begin{document}
\title{Cox reduction and confidence sets of models: a theoretical elucidation}
\author{Rebecca M. Lewis\footnote{Corresponding author, rml15@ic.ac.uk, Supplementary Material available on request.}~ and Heather S. Battey}
\affil{Imperial College London}
\date{}

\maketitle

\begin{abstract}
For sparse high-dimensional regression problems, Cox and Battey \cite{battey2018,battey2017} emphasised the need for confidence sets of models: an enumeration of those small sets of variables that fit the data equivalently well in a suitable statistical sense. This is to be contrasted with the single model returned by penalised regression procedures, effective for prediction but potentially misleading for subject-matter understanding. The proposed construction of such sets relied on preliminary reduction of the full set of variables. While various possibilities could be considered for this, \cite{battey2017} proposed a succession of regression fits based on incomplete block designs. The purpose of the present paper is to provide insight on both aspects of \cite{battey2017}. For an unspecified reduction strategy, we begin by characterising models that are likely to be retained in the model confidence set, emphasising geometric aspects. We then evaluate possible reduction schemes based on penalised regression or marginal screening, before theoretically elucidating the reduction of \cite{battey2017}. We identify features of the covariate matrix that may reduce its efficacy, and indicate improvements to the original proposal. An advantage of the approach is its ability to reveal its own stability or fragility for the data at hand.
\end{abstract}


\section{Introduction}\label{secIntro}

In the context of regression with a large number $p$ of potential explanatory features on $n\ll p$ independent individuals, usual practice is to identify a single low-dimensional model. Motivated by scientific studies in which subject-matter understanding is sought rather than immediate predictive success, \cite{battey2017} argued that when several reasonable explanations are statistically indistinguishable, one should aim to specify a confidence set of models. The arguments for this are in our view compelling, yet there may be some flexibility in how the broad goal is operationalised. The following approach was suggested in  \cite{battey2017}. 

A \textit{reduction phase} starts with the full set of $p$ variables and aims to identify a much smaller set $\hat{\mathcal{S}}\subset \{1,\dots,p\}$ indexing variables with apparent explanatory power. The set $\hat{\mathcal{S}}$ is called the \emph{comprehensive model}. Following this reduction, the \textit{model assessment phase} assesses all sub-models of $\hat{\mathcal{S}}$ for their compatibility with the data, having taken necessary measures to prevent miscalibration due to double use of the data. The \emph{confidence set of models} $\mathcal{M}$ is an enumeration of those small sets of variables that fit the data equivalently well in a suitable statistical sense. 

The reduction phase as outlined in \cite{battey2017} is termed \emph{Cox reduction} throughout the present paper and proceeds as follows.

\begin{itemize}
	\item Arrange the $p$ variable indices at random either in a $k \times  k$ square or a $k \times  k \times k$ cube, where preferably $k \leq 15$. Extensions to four or more dimensions are possible. There is no loss of generality in assuming that $p$ is a perfect cube, as otherwise some positions are left unoccupied. A visualisation is in Figure S1 
	of the supplementary material \cite{lewisSupp}.
	\item Any intrinsic biological or physiological characteristics, known to be important, are either used to partition the sample or are included in all regressions and should not be arranged in the cube.
	\item Traverse the cube from its three directions, the rows, columns and ``corridors'', specifying $3k^{2}$ sets of $k$ variables. Fit a least squares regression to each set.
	\item From each regression, provisionally select a small number of variables. This might be the two variables with most significant effect, or all those variables, if any, whose Student $t$ statistics exceed a threshold. Cox and Battey \cite{battey2017} recommended that variables never selected or selected only once be discarded, except in the absence of strong prior counter evidence.
	\item The next step depends on the number of variables remaining. If this remains large (in excess of 30, say), a second phase is carried out, similar to the first, typically based on a square rather than a cube. The resulting set of variables from this two-stage reduction is $\hat{\mathcal{S}}$.
\end{itemize}
A form of this procedure is presented in pseudo-code in Section S1 
of the supplementary material \cite{lewisSupp}. 

The arrangement of variables into $3k^{2}$ distinct blocks of size $k$ is the partially balanced incomplete block design of Yates \cite{yates1936} used in a design setting where the number of treatments exceeds the number of experimental units per block. The context of Cox reduction is fundamentally different in that there is no blocking of units and no randomisation of treatments to units. The analogy is that the design aims to minimise loss or redundancy of information. As contrasted with, say, fitting $3k^2$ regressions on totally random sets of $k$ variables, Cox reduction ensures that every variable is assessed an equal number times, and never alongside the same set of variables. A practical advantage of Cox reduction is the reassurance it offers, through rerandomisation of the variable indices in the cube, over the security or otherwise of the conclusions. In particular, unstable sets $\hat{\mathcal{S}}$ point to fragility of the method on the data at hand and thereby guide the choice of tuning parameters. 

While inspiration came from \cite{yates1936}, motivation for Cox reduction is from Bradford Hill's \cite{BradfordHill1965} discussion of the circumstances under which an effect obtained in an observational study is relatively likely to have a causal interpretation.  Such conditions include that the effect is reproduced in independent studies and behaves appropriately when the potential cause is applied, removed and then reinstated \cite{cox1992}. See \cite[][p.165--6]{CoxDonnelly} for further discussion. These aspects point to Cox reduction as at least a broadly appropriate approach for an initial phase of analysis, to be considered alongside some more established procedures such as marginal screening or penalised regression. We will show in Section \ref{sec:confSet} that the set of variables retained by Cox reduction is better suited than some other reduction strategies for assessment of models in the second phase. A further motivation for theoretical elucidation is that the procedure appears to be in use, as indicated by the download logs for the R package \texttt{HCmodelSets} \cite{hhh}. Some published applications of Cox reduction include \cite{rubin}, \cite{CK3}, \cite{CK2}, \cite{CK1}.

Among the most relevant procedures proposed elsewhere in the high-dimensional literature is the multi-environment knockoff filter discussed by \cite{li2021}, although this generally comes with a different type of theoretical guarantee. The confidence set of \cite{battey2017} identifies a number of models that are by definition statistically indistinguishable, whereas the knockoff procedure returns a single model with a control on the false discovery rate \cite{barber2015, candes2018}.

In a low-dimensional context, confidence sets of models have been emphasised repeatedly by Cox (see e.g., \cite{cox1968}, \cite{CoxSnell1974}, \cite[Appendix A.2.5]{CoxSnell1989}). They have been considered from a different perspective in \cite{hansen2011}, which begins with a collection of candidate models and aims to identify those with the smallest expected loss at a given significance level. In \cite{hansen2011}, the relative explanatory power of pairs of models in the candidate set are compared, whereas \cite{battey2017} uses the comprehensive model as a reference set against which to gauge the adequacy of each submodel. If no submodel achieves comparable fit, the model confidence set of \cite{battey2017} is empty, whereas the set of \cite{hansen2011} contains several models with equally poor fit. Bayesian approaches to model selection and the resulting credible sets of models are in the same vein, with obvious operational differences. For example, the posterior distributions in \cite{mitchell1988} and \cite{george1993} assign non-zero weight to multiple models.

The purpose of the present paper is to provide theoretical insight into the construction of the model confidence set proposed by \cite{battey2017}, revealing situations that lead the procedure to partially fail, either by discarding a genuinely relevant variable in the reduction phase, or by discarding the true model at the model assessment phase. Several modifications to the original proposal are possible, of which some were mentioned in \cite{battey2018} but others are new to the present paper, having emerged from theoretical analyses contained herein. 

A presentational quirk of the paper is that a theoretical analysis of the model assessment phase of the procedure, relating to the model confidence set $\mathcal{M}$, is discussed first (Section \ref{sec:confSet}) conditional on a comprehensive model $\hat{\mathcal{S}}$ having already been isolated. This is to emphasise the possibility of using other reduction strategies besides Cox reduction for the construction of $\hat{\mathcal{S}}$. Such alternatives are discussed in Sections \ref{posReduction}. An analysis of Cox reduction is presented in Section \ref{coxReduction}.

Two sources of randomness will be used in the forthcoming analysis. When considering the model assessment phase (Section \ref{sec:confSet}), the randomness is induced through the generative model for the outcome. For the theoretical analysis of Cox reduction (Section \ref{coxReduction}), the outcome is treated as fixed and the only source of randomness comes from the arrangement of variables in the cube. In Section \ref{coxReduction}, it is convenient to think of the outcome $Y \in \mathbb{R}^n$ as generated according to the linear model $Y=X\theta^0+\epsilon$ where $X\in \mathbb{R}^{n \times p}$ is a design matrix, $\theta^0 \in \mathbb{R}^p$ is a sparse vector satisfying $\|\theta^0\|_0=s \ll p$ and $\epsilon$ consists of independent and identically distributed random entries with mean zero and variance $\sigma^2$. Such a modelling assumption, while convenient for sharpening terminology and notation, plays a negligible role in results of Section \ref{coxReduction}. A referee indicated p.\,106 of Scheff{\'e} \cite{scheffe1959}, which emphasises the role of randomisation to ensure statistical inferences from a notional normal-theory linear model are fair approximations under more realistic generative models. See also the comments on binary outcomes in Section \ref{secBinary}.

\section{A simple explicit example}

The following explicit example helps motivate the confidence set of models perspective. Data were generated according to the the model
\begin{eqnarray}\label{toy_example}Y=X\theta^0+\epsilon, \quad \quad \epsilon \sim N_n(\mathbf{0}_n, \mathbb{I}_{n\times n})\end{eqnarray}
where $\theta^0=(1,1,0,\dots,0)^T \in \mathbb{R}^p$ with $p=25$ and $n=100$. Since $p$ is relatively small, there is no need for the reduction phase; confidence sets can be constructed directly. Rows of $X$ were generated independently from a normal distribution of zero mean and covariance 
\[\Sigma=\begin{pmatrix}
1 & 0.99 & 0.99 & 0.5 &0.5 & \dots &0.5\\
0.99 & 1 & 0.99 & 0.5 &0.5 & \dots &0.5\\
0.99 & 0.99 & 1 & 0.5 &0.5 & \dots & 0.5\\
0.5 & 0.5 & 0.5 & 1 & 0.5 & \dots &0.5\\
\vdots & & &  & \ddots && \vdots \\
0.5&0.5&0.5&\dots&0.5&1&0.5\\
0.5&0.5&0.5&\dots&0.5&0.5&1\\
\end{pmatrix}.\]
This setting was constructed so that the covariates indexed by $\{1,2,3\}$ are statistically indistinguishable. A linear regression model fitted with a LASSO penalty by cross-validation  selected the single model indexed by covariate $\{2\}$. A likelihood ratio test of each model $\mathcal{S}_m \subseteq \{1,2,3\}$ against $[p]=\{1,\dots,p\}$ declared $\{2\}$, $\{1,2\}$, $\{2,3\}$ and $\{1,2,3\}$ as statistically indistinguishable from model $[p]$ at the $1\%$ level. These models constitute the confidence set $\mathcal{M}$, together with any other low-dimensional sub-models of $[p]$ not rejected at the same significance level. All four models produced similar predicted values (the correlations between their fitted values and those obtained from the full model were at least $0.97$). While for prediction the choice between well-fitting models is rather arbitrary, if understanding is the goal, an arbitrary choice among statistically indistinguishable models is a misleading portrayal of the information in the data, particularly if the different models have fundamentally different scientific interpretations. For more elaborate examples based on real data see Section \ref{simulations} and the empirical work in \cite{battey2017}. 

\section{Notation}\label{notation}
Let $X$ be an $(n\times p)$-dimensional matrix of covariate observations with rows $(x_i^{T})_{i=1}^{n}$ corresponding to $n$ observational units. The letter $i$ is reserved for indexing the transposed rows in this way, with other letters specifying columns. Let $Y=(Y_1,\ldots,Y_n)^T$ be the vector of outcomes. Define $\theta^0 \in \mathbb{R}^p$ to be the maximiser of an expected log likelihood function, whose sample version is
\[
\ell(\theta)=\textstyle{\sum_{i=1}^n} \ell_i(\theta, Y_i), \quad \quad \theta \in \mathbb{R}^p\]
where $\ell_i(\theta,Y_i)$ is a function of $x_i^T\theta$. Let $\mathcal{S}=\{j: \theta^0_j \neq 0\}$ be the set of signal variables and any indices not included in this set specify noise variables. We refer to $X\theta^0$ as the signal.

Upper-case and lower-case letters denote matrices and vectors respectively. With the context ensuring no ambiguity, capital letters are also used for sets. For a vector $u \in \mathbb{R}^p$, a matrix $X\in \mathbb{R}^{n \times p}$ and a set $\mathit{A} \subseteq \{1, \dots, p\}$, $u_\mathit{A}$ denotes the vector of entries of $u$ indexed by the set $\mathit{A}$. The columns of $X$ indexed by $\mathit{A}$ 
are written $X_\mathit{A}$ if $|\mathit{A}|>1$ and $x_\mathit{A}$  if $|\mathit{A}|=1$.  For any $a \in \mathit{A}$, $\mathit{A}_{-a}$ is the set $\mathit{A}\backslash\{a\}$ and $\mathit{A}^c$ its complement $\{1, \dots, p\}\backslash \mathit{A}$. Let $P_X=X(X^TX)^{-1}X^T$, i.e.~the projection matrix onto the column span of $X$. To avoid double subscripting we write $P_\mathit{A}=X_\mathit{A}(X_\mathit{A}^TX_\mathit{A})^{-1}X_\mathit{A}^T$. The vector $e_j$ for $j\in\mathbb{N}\backslash\{0\}$ denotes the $j$-th standard basis vector whose dimension will be clear from the context. Let $\|\cdot\|_0$, $\|\cdot\|_2$ and $\|\cdot\|_\infty$  be the $\ell_0$, $\ell_2$ and $\ell_\infty$ vector norms. When the argument is a matrix, $\|\cdot\|_2$ refers to the spectral norm and $\|\cdot\|_\infty$ to the infinity matrix norm given by the maximum absolute row sum of the matrix. The maximum eigenvalue of a square matrix $M$ is written $\lambda_{\text{max}}(M)$. The sub-Gaussian norm of a univariate random variable $Z$ is 
\[\|Z\|_{\psi_2}=\sup_{q \geq 1}q^{-1/2}(\mathbb{E}|Z|^q)^{1/q}.\]The sample correlation coefficient between two centred vectors $u$ and $v$ is written
 \[R(u,v)=\frac{u^Tv}{\|u\|_2 \|v\|_2}.\]
The sample multiple correlation coefficient between a centred vector $u$ and a matrix $X$ with centred columns is
\[R(u,X)=\max_{\alpha \in \mathbb{R}^p}\frac{u^TX\alpha}{\|u\|_2 \|X\alpha\|_2}=R(u, P_Xu).\] For two matrices $X_\mathit{A}$ and $X_B$ with centred columns,
\[R(X_\mathit{A}, X_\mathit{B})=\max_{\alpha, \beta}\frac{\alpha^TX_\mathit{A}^TX_\mathit{B}\beta}{\|X_\mathit{A}\alpha\|_2 \|X_\mathit{B}\beta\|_2}=\|P_\mathit{A}P_\mathit{B}\|_2,\]
where the final equality follows from the proof of Lemma S5 
in \cite{lewisSupp}.

\section{Model assessment phase}\label{sec:confSet}
\subsection{Framework}

Given a comprehensive model $\hat{\mathcal{S}}$, the model assessment phase enumerates all sub-models of $\hat{\mathcal{S}}$ that fit the data equivalently well. As suggested by \cite{battey2017}, this may be obtained by identifying all submodels of $\hat{\mathcal{S}}$ that are not rejected in a likelihood ratio test at a given significance level $\vartheta$.

For any $\mathcal{S}_m \subseteq \hat{\mathcal{S}}$, define $\mathcal{U}=\mathcal{S}_m^c \cap \hat{\mathcal{S}}$ and
\begin{eqnarray*}
\Theta_0^{(m)}&=& \{\theta \in \mathbb{R}^p:\,\theta_{\mathcal{S}_m^c}=0\},\\
\Theta_1^{(m)}&=& \{\theta \in \mathbb{R}^p:\,\theta_{\mathcal{U}} \neq 0, \, \theta_{\hat{\mathcal{S}}^c}=0\}.
\end{eqnarray*}
Let $\hat{r}$ and $r_m$ be the column ranks of $\widetilde{X}$ and $\widetilde{X}_m$, where $\widetilde{X}_m$ consists of the columns of $X$ indexed by $\mathcal{S}_m$. The likelihood ratio test of $H_0: \theta^0 \in \Theta_0^{(m)}$ against $H_1: \theta^0\in \Theta_1^{(m)}$ is
\begin{equation}\label{eqLRT}
\psi_m(Y,X)=\ind\{w(\mathcal{S}_m)\leq \chi^2_{\hat{r}-r_m}(1-\vartheta)\},
\end{equation}
where $\chi^2_{\hat{r}-r_m}(1-\vartheta) $ is the $1-\vartheta$ quantile of the $\chi^2$ distribution with $\hat{r}-r_m$ degrees of freedom and
\[
w(\mathcal{S}_m) = 2\left\{\sup_{\theta \in \Theta_1^{(m)}\cup \Theta_0^{(m)}}\ell(\theta) - \sup_{\theta \in \Theta_0^{(m)}}\ell(\theta)\right\}.
\]
A model $\mathcal{S}_m$ is included in the model confidence set $\mathcal{M}$ if $\psi_m(Y,X)=1$. To reduce the computational burden, it is reasonable to construct the confidence set of models by only testing submodels of $\hat{\mathcal{S}}$ of size less than $s^\#$ say, chosen independently of $n$. The uncheckable assumption $s^\# \geq s$ is needed, otherwise the true model is necessarily excluded. Algorithm 1 in Section S1 
of the supplementary material of \cite{lewisSupp} provides pseudo-code for the model assessment phase.

The likelihood ratio test has greater power to distinguish between some models than others. This section uncovers geometric features associated with models in the confidence set. This is considered under the assumptions that $\mathcal{S}\subseteq \hat{\mathcal{S}}$ where $\mathcal{S}$ is the set of signal variables, and $\hat{\mathcal{S}}$ and $\mathcal{M}$ are independent. It follows that $\mathcal{S}$ is contained in the model confidence set with probability converging to the chosen level $1-\vartheta$ as $n\rightarrow \infty$. We begin by describing the local behaviour of the test using Le Cam's formulation of asymptotic normality. We then extend the results by examining more general departures from the null hypothesis. 

For the rest of this section, let $\gamma^0$ be the entries of $\theta^0$ indexed by the comprehensive model $\hat{\mathcal{S}}$ and let $\widetilde{X}$ be the corresponding columns of $X$. Rearranging covariates if necessary, assume that $\gamma^0$ corresponds to the first few entries of $\theta^0$. The notation $\ell (\gamma^0, 0)$ is used to represent $\ell(\theta^0)$ when $\theta^0=(\gamma^0,0)$.

\subsection{Behaviour under contiguous alternatives}\label{contiguous}
Interpretation of the log-likelihood ratio test $\psi_m(Y,X)$ in general models is aided through appeal to Le Cam's formulation of local asymptotic normality \cite{LeCam1960}. Under standard parametric regularity conditions, notably that the statistical model is differentiable in quadratic mean at $\gamma^0$, \cite[Theorem~7.2]{vaart1998} gives the following local asymptotic expansion of the log likelihood, which holds under model $\hat{\mathcal{S}}$ or any submodel thereof:
\begin{equation}\label{eqLocalExp}
\ell(\gamma^0+h_n,0) - \ell(\gamma^0,0) = h^T Z_{n} - \frac{1}{2}h^T H_0 h + o_{p}(1).
\end{equation}
In equation \eqref{eqLocalExp}, $H_0$ is the Fisher information per observation, $Z_{n}=n^{-1/2}\nabla_\gamma \ell(\gamma^0,0)$ converges in distribution to a normal random vector of mean zero and covariance matrix $H_0$, and $n^{1/2}h_n \rightarrow h$. Sufficient conditions for differentiability in quadratic mean are given by \cite[Proposition~7.6]{vaart1998}. These are satisfied, for example, by most exponential family regression models. It follows from \eqref{eqLocalExp} that 
\[
\ell(\gamma^0+h_n,0) - \ell(\gamma^0,0) \stackrel{\gamma^0}{\longrightarrow} N\biggl(-\frac{1}{2}h^{T}H_0 h, h^{T}H_0 h\biggr),
\]
where $\stackrel{\gamma^0}{\longrightarrow}$ means converges weakly under the true distribution defined by $\theta^0=(\gamma^0,0)$.

Evans \cite{evans2020} generalised this result to allow two models, whose components are both in a $O(n^{-1/2})$ neighbourhood of the true distribution. This is a restriction to so-called contiguous alternatives, which in the context of equation \eqref{eqLRT} would imply that the model indexed by $\mathcal{S}_m$, if false, only omits variables whose associated signal strengths are dominated by the estimation error of the maximum likelihood estimator. This characterisation belongs to formal theory but indicates the directions from $\gamma^0$ in which the likelihood ratio test is expected to have relatively high or low power. As in \cite{evans2020}, introduce $h_n$ and $\widetilde{h}_n$ such that $h_n\rightarrow 0$, $\widetilde{h}_n\rightarrow 0$ and $n^{1/2}(h_n - \widetilde{h}_n)\rightarrow k$. Provided that the true distribution belongs to a model whose density with respect to an appropriate dominating measure is doubly differentiable in quadratic mean, in the sense of \cite[Definition~2.8]{evans2020}, then 
\begin{eqnarray}\label{eqLLR}
\nonumber	\ell(\gamma^0+h_n,0) - \ell(\gamma^0+\widetilde{h}_n,0) 	 &=& \frac{k^T}{\sqrt{n}}\nabla_\gamma \ell(\gamma^0+\widetilde{h}_n,0) - \frac{1}{2}k^T H_0 k + o_p(1) \\
	& \stackrel{\gamma^0}{\longrightarrow}&  N\biggl(-\frac{1}{2}k^{T}H_0 k, k^{T}H_0 k\biggr).
\end{eqnarray}
To interpret this result in general models, suppose that $k$ is a scalar multiple of a unit eigenvector $v$ of $H_0$ with associated eigenvalue $\lambda$, so that $k=a v$, and  $k^{T}H_0 k=(a\lambda)^2 v^{T}v = a^2\lambda^2$. Since eigenvectors of $H_0$ indicate orthogonal directions in which the Fisher information varies from highest to lowest, equation \eqref{eqLLR} shows that the false hypotheses in a local neighbourhood of $\gamma^0$ that are most likely to be excluded from the confidence set $\mathcal{M}$ are those for which $k$ coincides with directions of high curvature of the log likelihood function.

\subsection{Extensions to other alternatives}

The results in the previous section focus on models $\mathcal{S}_m$ where the unknown parameter $\gamma^0$ can be consistently estimated. Depending on the form of the log-likelihood, consistent estimation of $\gamma^0$ is not necessary for accurate characterisation of the log-likelihood $\ell(\gamma^0,0)$. Consider the normal theory linear model with covariate matrix $\widetilde{X}$ and known variance. Let $\widetilde{X}_m$ consist of the columns of $\widetilde{X}$ indexed by $\mathcal{S}_m$. Then,
\begin{eqnarray*}
w(\mathcal{S}_m)&=&\sigma^{-2}\|(P_{\widetilde{X}}-P_{\widetilde{X}_m})Y\|_2^2 \sim \chi^2_{\hat{r}-r_m}(\lambda),\\
 \lambda&=&\sigma^{-2}\|(I-P_m)\widetilde{X}\gamma^0\|_2^2
\end{eqnarray*}
when $\widetilde{X}$ has full rank and $P_m$ is the projection matrix onto $\widetilde{X}_m$. A model $\mathcal{S}_m$ is included in the model confidence set with probability converging to $1-\vartheta$ whenever $\|(I-P_m)\widetilde{X}\gamma^0\|_2$ converges to zero, that is, whenever the portion of the signal $\widetilde{X}\gamma^0$ that is orthogonal to the column span of $\widetilde{X}_m$ converges to zero. This result applies irrespective of how well $\gamma^0$ is estimated by $\mathcal{S}_m$.

To extend the results in the previous section, we consider likelihoods that depend on $\gamma$ only through $\widetilde{X}\gamma$ for $\gamma \in \mathbb{R}^{\hat{s}}$. The following result shows that models satisfying assumptions (a)-(d) below are included in the model confidence set with probability converging to $1-\vartheta$ as the sample size grows. 

\begin{proposition}\label{modConfSet}
Suppose the log-likelihood function $\ell(\gamma,0)$ depends on $\gamma$ only through $\eta=\widetilde{X}\gamma$ and write $\bar{\ell}(\eta)=\ell(\gamma,0)$. Let $\eta^0=\widetilde{X}\gamma^0$ and suppose there exist unique maximisers
\begin{eqnarray*}
\hat{\eta}_m&=&\text{argmax}_{\eta \in \text{Col-Sp}(\widetilde{X}_m)}\bar{\ell}(\eta)\\
\hat{\eta}&=&\text{argmax}_{\eta \in \text{Col-Sp}(\widetilde{X})}\bar{\ell}(\eta).
\end{eqnarray*}
Assume the following:
\begin{enumerate}
\item[(a)] Weak omitted signal: $\|(I-P_m)\eta^0\|_2=o(n^{-1/2})$,
\item[(b)] Predictive consistency under both models: $\|\hat{\eta}-\eta^0\|_3^3=O_P(n^{-1/2})$ and $\|\hat{\eta}_m-\eta^0\|_3^3=O_P(n^{-1/2})$.
\item[(c)] A local asymptotic expansion: for $M_n=o(1)$,
\begin{eqnarray*}
\sup_{ h \in B_n}|\bar{\ell}(\eta^0+h)-f(\eta^0, h)|&=&o_P(1)
\end{eqnarray*}
where $B_n=\{h \in \mathbb{R}^n: \|h\|_3\leq M_n\}$ and
\begin{eqnarray*}
f(\eta^0,h)&=&\bar{\ell}(\eta^0)+h^T\nabla_{\eta} \bar{\ell}(\eta^0)+\frac{1}{2}h^T\nabla_{\eta \eta}^2\bar{\ell}(\eta^0)h.
\end{eqnarray*}
 Further, $\|n^{-1/2}\nabla_{\eta} \bar{\ell}(\eta^0)\|_2$ and $\|\nabla_{\eta \eta}^2\bar{\ell}(\eta^0)\|_2$ are $O_P(1)$. 
\item[(d)] Asymptotic normality of the score function: there exists a matrix $\widetilde{Q}$ of full rank whose columns span the column space of $\widetilde{X}$ and the first $r_m$ columns, denoted $\widetilde{Q}_m$, span the column space of $\widetilde{X}_m$ where
\begin{itemize}
\item  $\max_{i=1}^n\|\widetilde{q}
_i\|_2=O(1)$ where $\widetilde{q}_i^T$ denotes the $i$-th row of $\widetilde{Q}$,
\item  the eigenvalues of $\widetilde{Q^T}\widetilde{Q}/n$ are asymptotically bounded above and away from zero,
\item the following asymptotic limits hold
\[ (\widetilde{Q}^TJ^0\widetilde{Q})^{-1/2}\widetilde{Q}^TU^0\overset{d}{\longrightarrow}N(0, \mathbb{I}),\]
\[ (\widetilde{Q}_m^TJ^0\widetilde{Q}_m)^{-1/2}\widetilde{Q}_m^TU^0\overset{d}{\longrightarrow}N(0, \mathbb{I}),\]
 with $U^0=\nabla_{\eta} \bar{\ell}(\eta^0)$ and $J^0=-\nabla_{\eta \eta}^2\bar{\ell}(\eta^0).$
\end{itemize}
\end{enumerate}
Then,
\[\mathbb{P}\{\psi_m(Y,X)=1\}\rightarrow 1-\vartheta\]
as $n \rightarrow \infty$. 
\end{proposition}

The required $\|\cdot\|_3$-consistency appearing in assumptions (b) and (c) of Proposition \ref{modConfSet} is unusual but arises naturally when considering the asymptotic expansion of the log-likelihood as a function of $\eta=\widetilde{X}\gamma$ instead of $\gamma$. Let $h=\widetilde{X}v$ with $n^{1/2}\|v\|_2=O(1)$. 
Standard arguments show that under regularity conditions, 
\begin{eqnarray*}
\bar{\ell}(\eta^0+h)&=&\ell(\gamma^0+v,0)\\
&=&\ell(\gamma^0,0)+v^T\widetilde{X}^T\nabla_{\eta}\bar{\ell}(\eta^0)+\frac{1}{2}v^T\widetilde{X}^T\nabla_{\eta\eta}^2\bar{\ell}(\eta^0)\widetilde{X}v+o_P(1)\\
&=&\bar{\ell}(\eta^0)+h^T\nabla_{\eta}\bar{\ell}(\eta^0)+\frac{1}{2}h^T\nabla_{\eta\eta}^2\bar{\ell}(\eta^0)h+o_P(1).\end{eqnarray*}
This is an asymptotic expansion of the form given in assumption (c) of Proposition \ref{modConfSet}. However, the condition $n^{1/2}\|v\|_2=O(1)$ is strong as there exist cases where $n^{1/2}\|v\|_2$ is unbounded and the asymptotic expansion remains valid. To avoid this, we obtain a similar expansion about $\eta^0$ by identifying conditions on $h$ for which
\[h^T\{\nabla_{\eta\eta}^2\bar{\ell}(\eta^0)-\nabla_{\eta\eta}^2 \bar{\ell}(\eta)\}h=o_P(1),\]
where with slight notational inaccuracy, $\eta$ lies on the line joining $\eta^0+h$ and $\eta^0$, and may differ in each entry of the matrix. In certain generalised linear models, this term is of order $\|h\|_3^3$ which motivates assumptions (b) and (c). For further details, see the proof of Lemma S10 
in the supplementary material of \cite{lewisSupp}.

Section S6 
of the supplementary material of \cite{lewisSupp} gives conditions under which the assumptions in Proposition \ref{modConfSet} are satisfied by canonical generalised linear models. These conditions are met by the linear and logistic regression models when $\mathcal{S}_m$ satisfies $\|(I-P_m)\widetilde{X}\gamma^0\|_2=o(n^{-1/2})$ and so those models that only exclude a small portion of the signal $\widetilde{X}\gamma^0$ will be included in the model confidence set. It is not necessary for $\widetilde{X}$ to be full rank; see the proof of Proposition \ref{modConfSet} for details. The expected size of the model confidence set is at least $M(1-\vartheta)$ where $M$ denotes the number of models $\mathcal{S}_m$ satisfying $\|(I-P_m)\widetilde{X}\gamma^0\|=o(n^{-1/2})$. If there are many sets of covariates that are highly correlated in sample with signal variables, then the model confidence set will contain a large number of models on average.

\section{Reduction phase: evaluation of possible strategies}\label{posReduction}
The construction of a model confidence set hinges on a preliminary reduction to a set $\hat{\mathcal{S}}$ of manageable size that contains the true model with probability converging to one as $n,p\rightarrow \infty$. Possible reduction strategies include penalised regression procedures such as the LASSO \cite{TibshiraniLasso1996} or marginal screening \cite{fan2008}. We briefly discuss some of the considerations involved, with an emphasis on the linear model. 

\subsection{Penalised regression}

Penalised regression performs variable selection by minimising the least-squares or negative log-likelihood function subject to a constraint on the magnitudes of the entries of the parameter vector. An estimate of $\theta^0$ is first obtained satisfying 
$\hat{\theta}^{(\lambda)}=\text{argmin}_{\theta \in \mathbb{R}^p} \{-\ell(\theta)+\lambda p(\theta)\},$
where $p(\theta)$ is typically of the form
\[p(\theta)=\sum_{j=1}^p p_j(|\theta_j|)\]
 and the set \begin{eqnarray}\label{lasso}\hat{\mathcal{S}}=\{j:\, \hat{\theta}^{(\lambda)}_j \neq 0\},\end{eqnarray}
 may be used to specify the comprehensive model. The tuning parameter $\lambda$ determines the size of the set $\hat{\mathcal{S}}$. The LASSO \cite{TibshiraniLasso1996}, SCAD \cite{fan2011} and MCP \cite{zhang2010} procedures arise from particular choices of $p(\theta)$.

As $p(\theta)$ is typically non-decreasing in the magnitudes $|\theta_j|$ for  $j \in \{1, \dots, p\}$, the comprehensive model obtained through penalised regression will rarely detect signal variables that are weakly correlated with the response variable. Decreasing the size of the tuning parameter introduces further covariates, however those that are correlated with the error term will be prioritised over those variables with a weak signal, resulting in an overfitted model. Cox reduction aims to avoid overfitting by performing many low-dimensional regressions and retaining only those variables that are consistently statistically significant at a chosen level. 

Penalised regression may also fail to select all signal variables even when they are highly correlated with the response variable. This occurrence is explained in \cite{efron2004} for the linear model with LASSO solution $\hat{\theta}^{(\lambda)}$ obtained using the LARS algorithm. Briefly, given a current predicted response, the algorithm sequentially adds the covariate that is most highly correlated with the residuals to update the predicted response. The sequence of solutions obtained at each step of the algorithm corresponds to the solutions of the LASSO problem for decreasing $\lambda$ \cite[Theorem~1]{efron2004}. When two covariates are highly but not perfectly correlated in sample, including one of them in a LARS step reduces the correlation between the other and the residual, thus making it difficult for the second variable to enter the model. This can lead the LASSO to select unimportant noise variables over signal variables. The simulation results summarised at the start of Section \ref{simulations} provide examples of this, showing that when signal strengths are weak and correlations among covariates are high, the comprehensive model determined by an undertuned LASSO often omits signal variables.  

The situation is in principle less problematic when there are perfectly correlated signal variables. If a LASSO solution includes at least one of the perfectly correlated signal variables, then there must exist another solution that includes all of these variables. This applies to general loss functions and other penalty functions too. Unfortunately, commonly used optimisation algorithms such as coordinate descent \cite{friedman2010, breheny2011} only return  one arbitrarily chosen solution, leaving open the possibility that some signal variables are discarded.

For these reasons, an undertuned penalised regression procedure is not recommended for construction of the comprehensive model.

\subsection{Marginal screening}\label{secMS}

 For sets $\mathit{E}$ and $\mathit{F}$, let 
\begin{eqnarray*}
\hat{\theta}_{Y:\mathit{E}}&=& (X_\mathit{E}^TX_\mathit{E})^{-1}X_\mathit{E}^TY\\
\hat{\theta}_{\mathit{F}:\mathit{E}}&=&(X_\mathit{E}^TX_\mathit{E})^{-1}X_\mathit{E}^TX_F
\end{eqnarray*}
and $\hat{\theta}_{Y:\mathit{E}.\mathit{F}}$ be the entries of $\hat{\theta}_{Y:\mathit{K}}$ corresponding to $\mathit{E}$ when $\mathit{K}=\mathit{E}\cup \mathit{F}$. The notation closely follows \cite{CoxWermuth1996}. Marginal screening \cite{fan2008} retains variables with the largest absolute marginal correlation with the outcome. To explore the limitations of marginal screening for construction of the comprehensive model, consider the  decomposition
	\begin{equation}\label{eqCochran}
\hat{\theta}_{Y:\mathit{E}}=\hat{\theta}_{Y:\mathit{E}.\mathit{F}}+(\hat{\theta}_{\mathit{F}:\mathit{E}})^T \hat{\theta}_{Y:\mathit{F}.\mathit{E}}.
\end{equation}
where $\mathit{E}, \mathit{F} \subset\{1, \dots,p\}$ with $|\mathit{E}|=1$. The result \eqref{eqCochran} was first noted in \cite{cochran1938} for sets $\mathit{F}$ of size one and has been used subsequently in \cite{cox1960, CoxWermuth2003}. A derivation for $|\mathit{F}|>1$ in the linear model is given in Section S4 
of the supplementary material of \cite{lewisSupp} (proof of Lemma S1). 
An asymptotic analogue of equation \eqref{eqCochran} for general regression models was given by \cite{CoxWermuth1990}. 
 It follows from \eqref{eqCochran} that $\hat{\theta}_{Y:\mathit{E}}=\hat{\theta}_{Y:\mathit{E}.\mathit{F}}$ if and only if either $\hat{\theta}_{\mathit{F}:\mathit{E}}=0$ or $\hat{\theta}_{Y:\mathit{F}.\mathit{E}}=0$ or the two vectors are orthogonal. In particular, if this condition is violated when $\mathit{E}$ indexes a signal variable and $\mathit{F}=\mathcal{S} \backslash \mathit{E}$, it is possible for the marginal effect to be inflated or diminished. This improves or curtails the ability of marginal screening to detect $\mathit{E}$. If there is total cancellation, i.e. $\hat{\theta}_{Y:\mathit{E}.\mathit{F}} \approx - (\hat{\theta}_{\mathit{F}:\mathit{E}})^T\hat{\theta}_{Y:\mathit{F}.\mathit{E}}$, marginal screening would be unable to detect the signal variable indexed by $\mathit{E}$.

The situation described above is also sometimes challenging for Cox reduction, as will become clear in Section \ref{coxReduction}. Differences between the two procedures are most apparent when there is partial cancellation, so that the signal variable indexed by $\mathit{E}$ is deemed ineffective in the relatively strong reduction effectuated by marginal screening, but survives the weaker first round of Cox reduction, giving it the opportunity to be assessed in the presence of other strong variables in the second round of reduction.

A more subtle point is that the noise variables retained by Cox reduction facilitate model discrimination at the model assessment phase to a greater extent than those retained by marginal screening. The explanation is that marginal screening retains covariates that are highly correlated in sample with $Y$, and so the variables included in the comprehensive model generally span a relatively low-dimensional subspace of $\mathbb{R}^n$. Since Cox reduction requires that any apparent effect is not explained away by the companion variables, the angles between $Y$ and the resulting $n$-dimensional vectors of observations on retained noise covariates need not be small. The implication is that $X_{\hat{\mathcal{S}}}$ typically spans a larger subspace of $\mathbb{R}^n$. In this way, Cox reduction is expected to identify a comprehensive model that fits the data better than that identified by marginal screening, making it harder for submodels to pass the likelihood ratio test. Furthermore, submodels of the comprehensive model obtained from marginal screening span a similar space to $X_{\hat{\mathcal{S}}}$ by construction, and so are unlikely to be rejected by a likelihood ratio test (see Section \ref{sec:confSet}).

While the advantages of Cox reduction over marginal screening in particular circumstances are intuitively clear based on \eqref{eqCochran}, there are also situations in which marginal screening outperforms Cox reduction. For this reason, Section \ref{coxReduction} studies a setting in which marginal screening would be an obvious candidate for the construction of $\hat{\mathcal{S}}$. This is to highlight limitations of Cox reduction and point to potential improvements. 
\section{Some theoretical analysis of Cox reduction}\label{coxReduction}
\subsection{Framework}\label{framework}
Cox reduction forms a comprehensive model by analysing variables alongside randomly selected sets of companion variables and retaining those whose apparent effects are not explained away by other variables. Some analysis of Cox reduction was provided in \cite{battey2018}, where the focus was on the comparison of decision rules, legitimising certain simplifying assumptions. Based on their results, we consider the following version of Cox reduction.
\begin{itemize}
\item In the first round of Cox reduction, variable indices are randomly arranged in a $k\times k \times k$ cube. Variables are retained for further analysis if they are among the two most significant in at least two out of the three regressions in which they appear. See Algorithm 2 in Section S1 
of the supplementary material \cite{lewisSupp}.
\item In the second round of Cox reduction, the indices corresponding to variables retained after the first round are randomly arranged in a $k\times k $ square. The comprehensive model indexes those variables that are significant at level $\alpha$ in at least one regression. See Algorithm 3 in Section S1 
of the supplementary material \cite{lewisSupp}.
\end{itemize}

The present section identifies potentially problematic situations that lead Cox reduction to retain noise variables over signal variables. For this analysis, we treat $Y$ and $X$ as fixed, with the only source of randomness coming from the arrangement of variable indices in the cube. Each of the regressions that are fitted are linear with variable significance determined by the magnitude of a Wald statistic. This is most natural when the generative model is of the form $Y=X\theta^0+\epsilon$, where $\epsilon$ consists of centred, independent entries with known variance $\sigma^2$. However, with the exception of Proposition \ref{lem:maxSpurCorr}, the conclusions are free of modelling assumptions. 

\subsection{Assumptions}
We make the following assumptions on the design matrix and response vector. This will allow the relative significance of variables to be approximated by more interpretable quantities.

\begin{condition}\label{ch2CRCond}
Suppose $Y$ and all columns of $X$ are centred, and $X_{\mathit{K}}$ is of full rank for all $\mathit{K}\subset\{1,\dots,p\}$ of cardinality $k$.
\end{condition}

\begin{condition}\label{ch2CRCond2}
Suppose the indices $\{1, \dots, p\}$ can be partitioned into disjoint sets $\mathcal{A}$ and $\mathcal{B}=\{1, \dots, p \} \backslash \mathcal{A}$ such that
\begin{eqnarray*}
\min_{\mathit{A}\subseteq \mathcal{A},\, |\mathit{A}|=1}\{R(Y, x_\mathit{A})\}^{-1}=O(1)\label{eq:marg}
\end{eqnarray*}
and the maximum sample correlation 
\begin{eqnarray*}
\Delta(a):=\max_{\substack{\mathit{A} \subseteq \mathcal{A},\,\mathit{B} \subset \mathcal{B},\\ |\mathit{A} \cup \mathit{B}|=k, |\mathit{A}|\leq a}}\max\{R^2(X_\mathit{A}, X_\mathit{B}), R^2(Y, X_\mathit{B})\}
\end{eqnarray*}
satisfies $\Delta(a)=o(1)$ when $a$ is of moderate size (to be specified where necessary). When $|\mathit{A}|=0$ or $|\mathit{B}|=0$, we let $R^2(X_\mathit{A},X_\mathit{B})=0$ or $R(Y, X_\mathit{B})=0$ so that $\Delta(a)$ is well-defined.
\end{condition}

Given our decision to treat $Y$ and $X$ as fixed, the terms $\Delta(a)$ or $R(Y,x_A)$ are, for present purposes, non-random quantities that vary as $n,p \rightarrow \infty$. Condition \ref{ch2CRCond} is minor and ensures the Wald statistics are well-defined and can be related to sample correlation coefficients as defined in Section \ref{notation}. To interpret and justify Condition \ref{ch2CRCond2}, it will be useful to temporarily view the rows of the design matrix as independent and identically distributed random quantities with $\mathcal{A}$ indexing signal variables and noise variables correlated in population with signal variables, and $\mathcal{B}$ indexing the remaining noise variables. Under suitable conditions on the random process, $R(Y,x_A)$ is a sample correlation between two quantities that are related at the population level, and so will likely be bounded away from zero asymptotically. In contrast, $R(X_\mathit{A}, X_\mathit{B})$ and $R(Y, X_\mathit{B})$ are sample correlations between unrelated random variables and so these are expected to decay to zero as $n \rightarrow \infty$ provided the cardinality of the sets $\mathit{A}$ and $\mathit{B}$ are of moderate size. This behaviour coincides with that required in Condition \ref{ch2CRCond2}. We focus on sets such that $|A\cup B|=k$, this being the block size in each traversal of the cube or square. The quantity $a$ in the definition of $\Delta(a)$ reflects a restriction on the number of variables in $\mathcal{A}$ appearing in each regression, and we will soon see that during the first round of Cox reduction, it is sufficient to consider $a=1$. Based on this interpretation, we refer to $\mathcal{A}$ as the set of \textit{pseudo-signal variables} and $\mathcal{B}$ as the set of \textit{pseudo-noise variables}.

Proposition \ref{lem:maxSpurCorr} provides a more precise justification for Condition \ref{ch2CRCond2} in the case $a=1$ when the entries of the design matrix and response vector are random observations from sub-Gaussian distributions. The result shows that $\Delta(1)=o_P(1)$ by making use of Theorem 3.1 in \cite{fan2018}. The behaviour of $\Delta(a)$ for $a>1$ was not considered by \cite{fan2018}.

\begin{proposition}\label{lem:maxSpurCorr} 
Assume the following:
\begin{itemize}
\item Each entry of $Y$ is an independent observation of a centred sub-Gaussian random variable,
\item For $a\in \mathcal{A}$, each entry of $x_a$ is an independent observation of a centred sub-Gaussian random variable.
\item Each row of $X_\mathcal{B}$ is an independent sample from the distribution of $U$ where $U\in \mathbb{R}^{|\mathcal{B}|}$ is a random vector with independent and centred entries satisfying 
\[\sup_{\alpha \in \mathbb{R}^{|\mathcal{B}|}: \|\alpha\|_2=1}\|\alpha^TU\|_{\psi_2}<\infty.\]
\item $X_{\mathcal{A}}$ and $Y$ are both independent of $X_{\mathcal{B}}$ (but possibly dependent on each other).
\end{itemize}
Then,
\[\Delta(1)=O_P\{\log(|\mathcal{B}|)/n\}=o_P(1)\]
as $p,n\rightarrow \infty$ with $k, k^{-1}=O(1)$, $\log(|\mathcal{B}|)=o(n)$, $|\mathcal{A}|=o\{ \log (|\mathcal{B}|)\}$ and $|\mathcal{A}|^{8/7} \log\{|\mathcal{B}|n\}=o(n^{1/7})$.
\end{proposition}

When $|\mathcal{A}|=O(1)$ and $|\mathcal{B}|=O(p)$, Proposition \ref{lem:maxSpurCorr} shows that $\Delta(1)=o_P(n^{-6/7})=o_P(1)$ as long as $\log p=o(n^{1/7})$. When this assumption is violated, it is possible for noise variables to be retained instead of signal variables. This is discussed further in Section \ref{problems} and some of the modifications outlined in Section \ref{improvement} are designed to mitigate the issue.

Our formulation is deliberately constructed so that when $\mathcal{A}$ includes all signal variables, Cox reduction offers no obvious advantage over marginal screening as far as retention of signal variables is concerned. Limitations of Cox reduction may then be identified and potential improvements proposed. 

\subsection{Preliminary insights}\label{sec:prelim}

For a set $\mathit{K}\subset \{1, \dots ,p\}$ of size $k$ indexing variables in a given $k$-dimensional regression, suppose Cox reduction retains variables according to the Wald statistic
\[
T_{Y:\mathit{K}}=\sigma^{-1}D_{\mathit{K}}^{-1/2}\hat{\theta}_{Y:\mathit{K}}
\]
where $D_{\mathit{K}}$ is a diagonal matrix with entries given by the diagonal entries of $(X_{\mathit{K}}^TX_{\mathit{K}})^{-1}$ and $\hat{\theta}_{Y:\mathit{K}}$ was defined in Section \ref{secMS}. Lemma S2 
in \cite{lewisSupp} shows that the entry of $T_{Y:\mathit{K}}$ corresponding to variable index $e \in \mathit{K}$ is
\begin{eqnarray}\sigma^{-1}\|Y\|_2 R(Y, (I-P_{\mathit{K}_{-e}})x_e)\label{eq:tStatResultText}
\end{eqnarray}
and so the relative sizes of the Wald statistics depend on the sample correlation between the response variable $Y$ and the projected covariates $(I-P_{\mathit{K}_{-e}})x_e$.  If a variable has a real effect on the response, this correlation will be large for many companion sets $\mathit{K}_{-e}$. Cox reduction is designed to exploit this. For conciseness, correlation in the present section means sample correlation unless indicated otherwise.

For a more precise analysis, let $\mathit{A} \subseteq \mathcal{A}$ and $\mathit{B}\subseteq\mathcal{B}$ index the variables appearing in a given traversal of the hypercube with $|\mathit{A}\cup \mathit{B}|=k$. Let $T_{Y:\mathit{A}.\mathit{B}}$ be the entries of $T_{Y: \mathit{K}}$ corresponding to $\mathit{A}$ when $\mathit{K}=\mathit{A}\cup \mathit{B}$. 
The following result allows us to characterise the limiting behaviour of the statistics $T_{Y:\mathit{A}.\mathit{B}}$ and $T_{Y:\mathit{B}.\mathit{A}}$ under Conditions \ref{ch2CRCond} and \ref{ch2CRCond2}.

\begin{proposition}\label{tStatApproxV2_comb}
Suppose $R(X_\mathit{A}, X_{\mathit{B}})\leq 1-c$ for some $c>0$. Then there exists $C>0$, depending only on $c$, such that 
\begin{eqnarray*}T_{Y:\mathit{A}.\mathit{B}}&=&T_{Y:\mathit{A}}+\delta_{\mathit{A}.\mathit{B}}\\
T_{Y:\mathit{B}.\mathit{A}}&=&T_{(I-P_\mathit{A})Y:\mathit{B}}+\delta_{\mathit{B}.\mathit{A}}^{(2)}\\
&=&\delta_{\mathit{B}.\mathit{A}}^{(1)}+\delta_{\mathit{B}.\mathit{A}}^{(2)}
\end{eqnarray*}
where 
\begin{eqnarray*}
\|\delta_{\mathit{A}.\mathit{B}}\|_{\infty}&\leq &C\sigma^{-1}\|Y\|\{R(X_\mathit{A}, X_B)+R(Y, X_\mathit{B})\},\\
\|\delta_{\mathit{B}.\mathit{A}}^{(1)}\|_{\infty}&\leq &\sigma^{-1}\|Y\|\{R(X_\mathit{A}, X_B)+R(Y, X_\mathit{B})\}\\
\|\delta_{\mathit{B}.\mathit{A}}^{(2)}\|_{\infty}&\leq & C\sigma^{-1}\|(I-P_\mathit{A})Y\|_2R(X_\mathit{A}, X_\mathit{B})\}.
\end{eqnarray*}
\end{proposition}
Proposition \ref{tStatApproxV2_comb} shows that as $R(X_\mathit{A}, X_\mathit{B}) \rightarrow 0$ and $R(Y, X_\mathit{B}) \rightarrow 0$, the entries of the scaled statistic \[\|Y\|^{-1}T_{Y:\mathit{A}.\mathit{B}}\] converge to the corresponding entries of the scaled $t$-statistic obtained by regressing $Y$ on $X_\mathit{A}$ only. In contrast, the entries of $\|Y\|^{-1}T_{Y:\mathit{B}.\mathit{A}}$ converge to the entries of the scaled $t$-statistic obtained by regressing $(I-P_\mathit{A})Y$ on $X_\mathit{B}$. The latter is negligible compared to the magnitudes of the entries of $\|Y\|^{-1}T_{Y:\mathit{A}.\mathit{B}}$. These observations will be of use when analysing each round of reduction in the following sections.

\subsection{First reduction}\label{R1}
In the first round of Cox reduction, variable $e \in \{1, \dots, p\}$ is assessed alongside variables indexed by $\mathit{K}_1$,  $\mathit{K}_2$ and  $\mathit{K}_3$, where $\mathit{K}_1\cap\mathit{K}_2\cap \mathit{K}_3=\{e\}$. Let $\mathcal{E}_{e: \mathit{K}}$ be the event that the absolute entry of $T_{Y:\mathit{K}}$ corresponding to $e \in \mathit{K}$ is among the two largest values in the set 
\[\{|e_j^TT_{Y:\mathit{K}}|\, : \, j \in \{1,\dots,k\}\}.\]
Then the variable indexed by $e$ is retained through the first reduction if at least two of the three events $\mathcal{E}_{e: \mathit{K}_1}$, $\mathcal{E}_{e: \mathit{K}_2}$ and $\mathcal{E}_{e: \mathit{K}_3}$ occur. See Algorithm 2 in Section S1 
of the supplementary material \cite{lewisSupp}.

When the number of pseudo-signal variables is small relative to the dimension $p=k^3$, it is likely that pseudo-signal variables are unaccompanied by other variables from $\mathcal{A}$ in at least 2/3 of the regressions in which they appear. By Proposition \ref{tStatApproxV2_comb}, pseudo-signal variables are retained through the first reduction based on their marginal relationship with the response. The following result establishes that as long as the marginal correlation between the response variable and a given pseudo-signal variable is sufficiently large, Cox reduction retains the pseudo-signal variable with probability close to one, where the randomness comes only from the arrangement of variable indices in the hypercube, the responses and covariates being treated as fixed at their realised values.  

\begin{figure}
\center
\includegraphics[trim={100 240 100 240}, clip, scale=0.6]{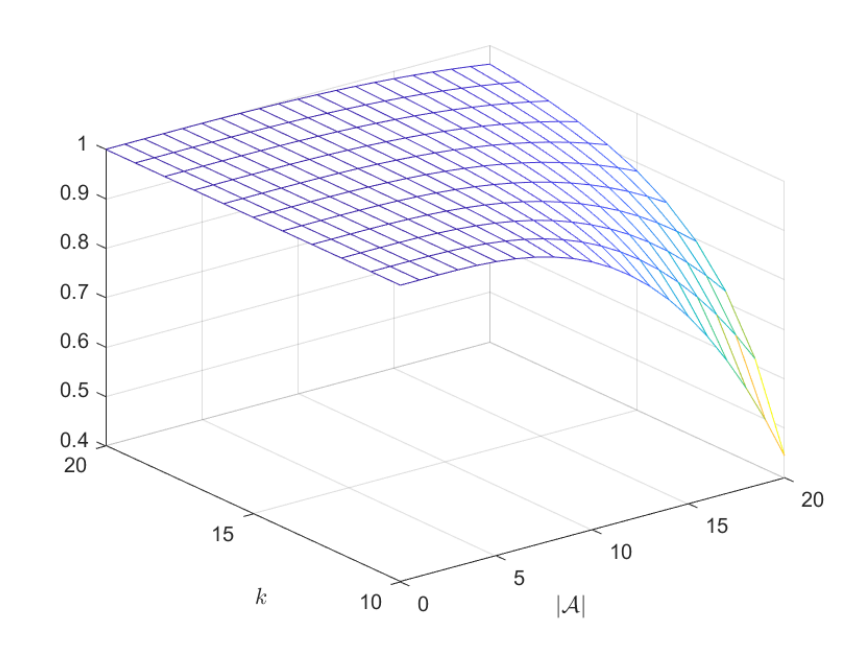}
\caption{Plot of equation \eqref{probR1} for various values of $k$ and $|\mathcal{A}|$. }\label{probRound1}
\end{figure}

\begin{proposition}\label{R1Result}
Suppose Conditions \ref{ch2CRCond} and \ref{ch2CRCond2} hold with $a=1$ when $p,n\rightarrow \infty$ with $k=O(1)$. Then, for $n$ large enough, the covariates indexed by $\mathcal{A}$ are retained after the first round of reduction with probability at least 
\begin{eqnarray}1-\frac{|\mathcal{A}|(|\mathcal{A}|-1)(|\mathcal{A}|-2)(k-1)^2}{(k^3-1)(k^3-2)}.\label{probR1}\end{eqnarray}
\end{proposition}

The lower bound on the probability given in Proposition \ref{R1Result} is a lower bound on the probability that pseudo-signal variables appear unaccompanied by other pseudo-signal variables in at least 2/3 regressions in which they appear. This bound is plotted in Figure \ref{probRound1}. When $|\mathcal{A}|$ is of moderate size relative to $p=k^3$, the probability is close to one. In light of this result, highly correlated variables may be paired, with one representative of both, during the first round of reduction. This  serves to reduce the cardinality of the the set $\mathcal{A}$, thus increasing probability \eqref{probR1}.

The result in Proposition \ref{R1Result} is conservative in that it focuses on the case where each pseudo-signal variable is the most significant in at least 2/3 regressions in which it appears. Cox reduction retains a variable if it is among the two most significant. 

The first reduction retains all  pseudo-signal variables whose marginal correlation with the response is sufficiently large and so mimics a conservative version of marginal screening that retains covariates with the largest $\hat{s}$ marginal correlations with the response, for $\hat{s}$ large. The two procedures differ in the set of retained pseudo-noise variables. The second round of Cox reduction, at which point the two procedures diverge more substantially, considers the joint explanatory power of sets of pseudo-signal variables. 

\subsection{Second reduction}
Suppose that the index $e$, retained through the first reduction, is randomly arranged in a $k \times k$ square, where the dimension $k$ may be different from that of the first reduction. Let $\mathit{K}_1$ and $\mathit{K}_2$ be the variables that share a row or column with $e$ in the square and redefine $\mathcal{E}_{e: \mathit{K}}$ to be the event that the entry of $T_{Y:\mathit{K}}$ corresponding to $e \in \mathit{K}$ is significant at level $\alpha$. Then, the second round of reduction retains the variable indexed by $e$ on the event $\mathcal{E}_{e: \mathit{K}_1} \cup \mathcal{E}_{e: \mathit{K}_2}$. See Algorithm 3 in Section S1 
of the supplementary material \cite{lewisSupp}.

In contrast to the first round, pseudo-signal variables are likely to appear together in second-round regressions. Proposition \ref{expecRound} derives the expected number of pseudo-signal variables that share a row or column with a given $a \in \mathcal{A}$ in terms of the dimension of the square $k$ and the size of the set $\mathcal{A}$. For comparison, the analogous calculation for the cube is also provided. Figure \ref{plotExpecR1R2} shows that this expectation is substantially larger in the second round. 

\begin{proposition}\label{expecRound}
Suppose the indices in $\mathcal{A}$ are randomised to positions in a $k\times k $ square. For a given index $a \in \mathcal{A}$, the expected number of indices from $\mathcal{A}_{-a} $ sharing a row or column with $a$ is 
\begin{eqnarray*}
\frac{2(|\mathcal{A}|-1)}{k+1}.\label{eq:expecR2}
\end{eqnarray*}
In contrast, if the indices are randomised to positions of a $k \times k\times k$ cube, the expected number of indices from $\mathcal{A}_{-a}$ sharing a row, column or corridor with $a$ is 
\begin{eqnarray*}
\frac{3(|\mathcal{A}|-1)}{k^2+k+1}.\label{eq:expecR1}
\end{eqnarray*}
\end{proposition}

\begin{figure}
\center
\begin{subfigure}{0.5\textwidth}
\center
\includegraphics[trim={100 270 100 280}, clip, scale=0.6]{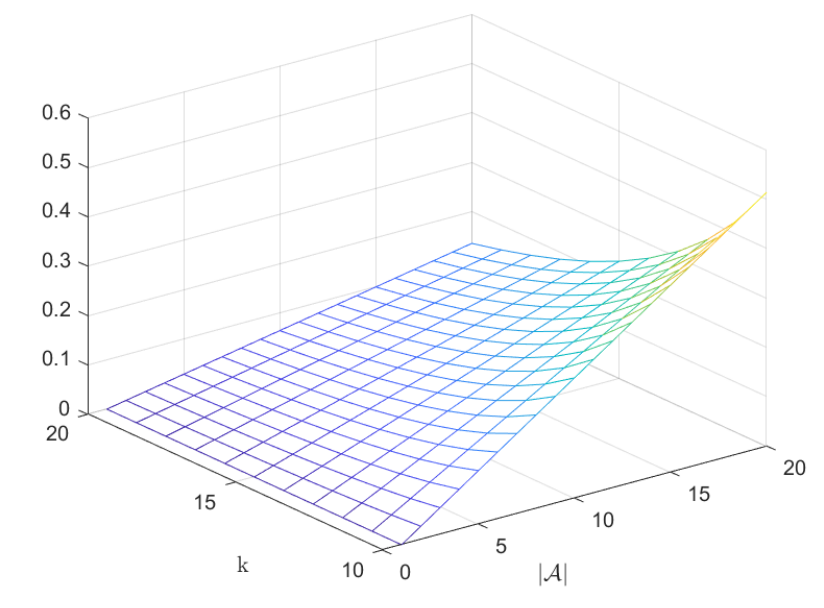}
\caption{ Round 1: $k \times k \times k$ cube}
\end{subfigure}
\hfill
\begin{subfigure}{0.5\textwidth}
\center
\includegraphics[trim={100 270 100 280}, clip, scale=0.6]{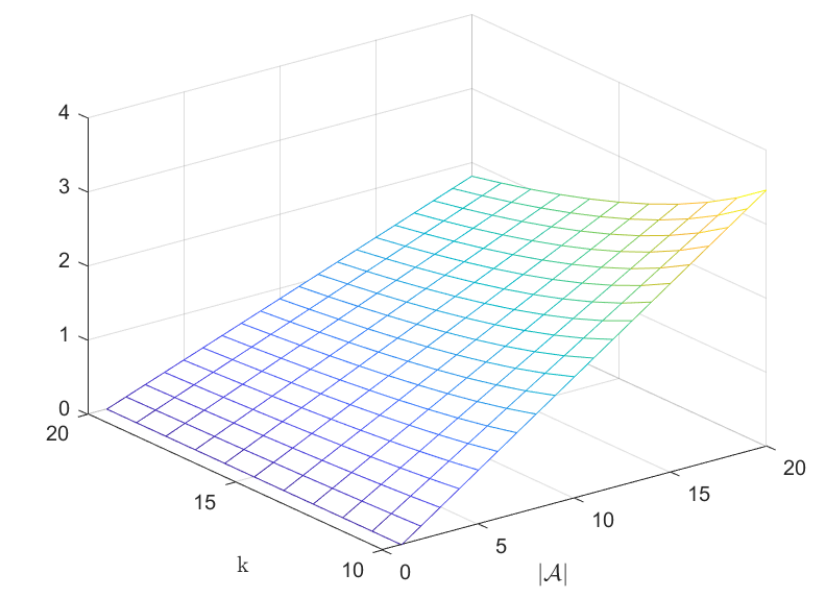}
\caption{Round 2: $k \times k$ square}
\end{subfigure}
\caption{Expected number of indices from $\mathcal{A}_{-a}$ sharing a fibre with $a$ for various values of $k$ and $|\mathcal{A}|$.}\label{plotExpecR1R2}
\end{figure}

Provided Conditions \ref{ch2CRCond} and \ref{ch2CRCond2} hold for $a=k$, Propositions \ref{tStatApproxV2_comb} and \ref{expecRound} imply that pseudo-signal variables are retained through the second reduction based on the significance of statistics of approximate form
\begin{eqnarray}\sigma^{-1}\|Y\|_2R(Y, (I-P_{\mathit{A}_{-a}})x_a)\label{round2stat}\end{eqnarray}
where $\mathit{A}\subseteq\mathcal{A}$ denotes the pseudo-signal variables appearing in a given row or column and $a \in \mathit{A}$. Since the second reduction compares the magnitude of these statistics to a pre-defined threshold, the retention of pseudo-signal variables is approximately independent of all pseudo-noise variables. Provided the dependence between pseudo-signal variables does not make $R(Y, (I-P_{\mathit{A}_{-a}})x_a)$ small for many possible sets $\mathit{A}$ and indices $a$, Cox reduction will retain pseudo-signal variables.

In contrast, any pseudo-noise variables that survived the first reduction are retained through the second based on the magnitudes of statistics of approximate form
\[\sigma^{-1}\|(I-P_\mathit{A})Y\|_2R\{(I-P_\mathit{A})Y, (I-P_{\mathit{B}_{-b}})x_b\}\]
where $\mathit{B}\subseteq\mathcal{B}$ and $\mathit{A}\subseteq\mathcal{A}$ denote the pseudo-noise and pseudo-signal variables appearing in the same regression as $b \in \mathit{B}$. These statistics are bounded in magnitude by 
\begin{eqnarray*}
 \sigma^{-1}\|Y\|_2\Delta(|\mathit{A}|)^{1/2}
\end{eqnarray*}
up to a constant and so their relative sizes compared to (\ref{round2stat}) are asymptotically small. Cox reduction is not expected to prioritise these indices for retention. Proposition \ref{resultR2} formalises these observations, showing that when the sample size is large enough, pseudo-signal variables included in a set $\mathcal{A}_1\subseteq \mathcal{A}$ are ranked more highly by Cox reduction than pseudo-noise variables. Provided the threshold for retention is chosen appropriately, Cox reduction will succesfully distinguish between these two types of variables. 

\begin{proposition}\label{resultR2}
Suppose Conditions \ref{ch2CRCond} and \ref{ch2CRCond2} hold with $a=k$.  Define $\mathcal{A}_1\subseteq\mathcal{A}$ to satisfy
\[\max_{a \in \mathcal{A}_1}\max_{\mathit{A} \subseteq \mathcal{A}\backslash\{a\},\, |A|\leq k-1}\{|R(Y, (I-P_{\mathit{A}})x_a)|\}^{-1}=O(1).\]
Then there exist $N,\alpha>0$ such that when $n\geq N$ the indices in $\mathcal{A}_1$ are retained by the second reduction but those in $\mathcal{B}$ that survived the first reduction, are not. This holds irrespective of the random arrangement of indices in the hypercube.
\end{proposition}

It is possible that $\mathcal{A}_1$ is empty or does not include the full set of pseudo-signal variables. In this case, the variables in $\mathcal{A}\backslash\mathcal{A}_1$ may be retained by Cox reduction, although this will depend on the random arrangement of the indices in the hypercube.

The second reduction requires any apparent explanatory power to persist when the variable in question is accompanied by relatively strong companion variables. Its intended purpose was to differentiate variables in $\mathcal{S}$ from those in $\mathcal{S}^{c}\cap \mathcal{A}$, when $\mathcal{S}\subseteq \mathcal{A}$. Consider a regression of $Y$ on $X_\mathit{A}$ where $\mathit{A}\subseteq \mathcal{A}$ contains signal and noise indices. Lemma S3 
in \cite{lewisSupp} shows that the entry of $\sigma T_{Y:A}$ corresponding to $a \in \mathit{A}$  is given by
\begin{eqnarray}\Delta_{1}+\Delta_{2}\label{eq:tStat_NS}
\end{eqnarray}
where 
\begin{eqnarray*}
\Delta_1&=&\sqrt{1-R^2(x_a, X_{\mathit{A}_{-a}})}\theta^0_a\|x_a\|_2\\
\Delta_2&=&\|Y-X_\mathit{A}\theta^0_\mathit{A}\|_2R(Y-X_\mathit{A}\theta^0_\mathit{A}, (I-P_{\mathit{A}_{-a}})x_a).
\end{eqnarray*}
The term $\Delta_1$ describes the amount of the true signal corresponding to each variable that is recovered, and is equal to zero when $a$ indexes a noise variable. The term $\Delta_2$ characterises how well the variables in $A$ are able to recover the omitted signal $Y-X_\mathit{A}\theta^0_\mathit{A}$ and is bounded in magnitude by $\|Y-X_\mathit{A}\theta^0_\mathit{A}\|_2$. A regression of $Y$ on $X_\mathit{A}$ is able to distinguish between the noise and signal variables in $\mathit{A}$ if the term $\min_{a \in \mathcal{S}\cap \mathit{A}}|\Delta_1|$ is large in comparison to $\max_{a \in \mathit{A}}|\Delta_2|$.
 If $a$ indexes a noise variable, then $\Delta_1=0$ and $|\Delta_2|$ is minimised (up to terms involving $\epsilon$) when $\mathit{A}$ includes the signal variables that are not orthogonal to $x_a$. In contrast, $\Delta_1$ is largest for $a \in \mathcal{S} \cap \mathit{A}$ when colinearity between signal variables and other variables in $\mathit{A}$ is absent. Whether or not Cox reduction accurately differentiates variables in $\mathcal{S}$ from those in $\mathcal{S}^{c}\cap \mathcal{A}$ when $\mathcal{S}\subseteq \mathcal{A}$ therefore depends on the balance of colinearity among the variables present, and the proportion of signal that they contain.

In practice, the threshold $\alpha$ is chosen so that the set of retained variables is relatively stable across repeated rerandomisation of the indices in the two hypercubes. In view of the observations above, the set of pseudo-noise variables retained due to high spurious correlation $R((I-P_\mathit{A})Y, (I-P_{\mathit{B}_{-b}})x_b)$ is expected to differ upon repeated rerandomisation of the variable indices in the hypercubes. In contrast, provided the dependence between pseudo-signal variables does not cause $R(Y, (I-P_{\mathit{A}_{-a}})x_a)$ to be uniformly small over many sets $\mathit{A}$ and indices $a \in \mathit{A}$, pseudo-signal variables will be retained by Cox reduction irrespective of the arrangement of variable indices. Thus, a relatively stable set of retained variables over repeated rerandomisation suggests presence of variables with genuine explanatory power. If no threshold $\alpha$ results in a stable set of manageable size, this is a warning against the procedure for the data at hand. 
\subsection{Problematic situations}\label{problems}

The analysis above points to situations that may lead Cox reduction to exclude relevant variables. For example, a signal variable with a weak sample correlation with the outcome is relatively unlikely to be retained by Cox reduction. Section \ref{secMS} discusses how such a variable might arise. Most variable selection procedures exclude the possibility that signal variables have weak marginal or partial correlations with the response. See Condition 3 in \cite{fan2008}, and the definition of partial faithfulness or Assumption 4 in \cite{buhlmann2010}.

In the second round of Cox reduction, variables from $\mathcal{A}$ appear together with frequency depending on the dimension of the square. When a signal variable has a weak partial correlation with the response given many different variables from $\mathcal{A}$, it may be declared unimportant in both of its second-round regressions and hence discarded by Cox reduction. Rerandomisation or a more systematic arrangement of the variable indices in the second reduction improve the situation. This is discussed in Sections \ref{improvement} and \ref{systematic}.

Another unfavourable situation, more problematic for marginal screening than for Cox reduction, is when a noise variable that is uncorrelated with all signal variables has a large sample correlation with the response purely by chance. Simulations in \cite{fan2008} pointed to large spurious correlations for large $p$. This was formalised in \cite{fan2018} which established the limiting distribution of the maximum spurious correlation. The distributions of other order statistics were not discussed. Any spurious noise variable would be contained in $\mathcal{A}$ and likely survive the first reduction. Depending on the random arrangement of variables in the square, it may also survive the second reduction. Similarly, noise variables may be retained when multiple correlation coefficients between uncorrelated variables are spuriously large.

In view of the model assessment phase, inclusion of noise variables is less problematic than omission of signal variables provided there are not so many as to make assessment of models practically infeasible. Sample splitting and rerandomisation, discussed in Section \ref{improvement}, reduce the survival probability of  noise variables.

\section{Recommended improvements to Cox reduction}\label{improvement}
\subsection{Alternating subsamples}\label{sampleSplit} 
We propose using a different portion of the sample for each reduction round, breaking the dependence between analyses in the first and second reductions and thereby reducing the probability that irrelevant noise variables appear correlated with the response in both rounds. Let $\mathcal{I}$ be a subset of $\{1, \dots, n\}$. The split-sample version of Cox reduction uses observations indexed by $\mathcal{I}$ for the first reduction and the model assessment phase, and those indexed by $\mathcal{I}^c$ for the second reduction. For more than two rounds, alternation should continue in this way. We recommend including $30-40\%$ of observations in $\mathcal{I}$ for the two-stage procedure, based on the relative difficulty of the stages. 

\subsection{Rerandomisation}\label{secRerandomisation}

A benefit of Cox reduction is the external source of randomness from the arrangement of variable indices in successive hypercubes. This provides a way of internally calibrating the procedure by rerandomising the variable indices. Battey and Cox \cite{battey2018} suggested this as a check on the stability of $\hat{\mathcal{S}}$. An elaboration is to rerandomise repeatedly, resulting in sets $\hat{\mathcal{S}}_1, \hat{\mathcal{S}}_2, \ldots$, and only retain variables in the final set $\hat{\mathcal{S}}$ if they are present most of the time, e.g.~in at least half the rerandomisation outcomes. With this adaptation, Cox reduction is less dependent on the particular set of variables appearing in regressions together, reducing the chances that signal variables are omitted due to small partial correlations or that noise variables are retained due to high spurious correlation.

 A natural and partially refutable criticism is that the procedure entails at least one tuning parameter (from the second-round reduction), increased to two when rerandomisation is introduced. In practice, selection of tuning parameters is guided by stability of the set $\hat{\mathcal{S}}$ and is rarely problematic. If there appears to be no choice of tuning pair that delivers a stable outcome, that is a warning against the procedure. The ability of Cox reduction to reveal its own fragility on the data at hand is an advantage over the alternatives discussed in Section \ref{posReduction}.
\section{Numerical performance}\label{simulations}
Section S2 
of the supplementary material \cite{lewisSupp} shows the numerical performance of Cox reduction and the proposed construction of the model confidence set. Simulated data were generated to compare marginal screening and LASSO penalised regression to Cox reduction in the setting laid out in Section \ref{coxReduction}. Cox reduction and marginal screening both performed well, retaining all signal variables in the comprehensive model and yielding a model confidence set with high simulated coverage probability. In contrast, the undertuned LASSO failed to include all signal variables in the comprehensive model when correlations among variables were high and signal strengths were low. The size of the model confidence set constructed from marginal screening contained more models on average than that obtained from Cox reduction, as expected based on the discussion in Section \ref{posReduction}. The recommended improvements to Cox reduction from Section \ref{improvement} performed favourably.

To test the performance of the procedure in the presence of more complicated dependence structures, an analysis based on real data was performed. These were from the online supplement of \cite{buhlmann2014} and consisted of the logarithm of the expression levels of $p=4088$ genes, measured for $n=71$ observational units. An artificial response was generated as
$Y=\mathbbm{1}_n+X\theta^0+\epsilon$, with $X$ from the genomics data, $\mathbbm{1}_n$ an $n$-dimensional vector of ones, and $\epsilon \sim N(0, 1)$. Four cases were considered to illustrate our theoretical insights by adjusting the set $\mathcal{S}$ of signal variables. In all cases $\theta^0_{\mathcal{S}}=2\mathbbm{1}_{3}$. 

To aid construction and understanding of the examples, we standardized the columns of $X$ to have unit sample standard deviation. This would not be appropriate in a genuine statistical analysis as it destroys the interpretation of the regression coefficients.

In each of 500 Monte Carlo replications an artificial response was generated anew and a confidence set of models was constructed using the most competitive of the reduction procedures from the analysis in Section S2 
of the supplementary material \cite{lewisSupp}, namely marginal screening (MS), Cox reduction with rerandomisation (CR-R) and Cox reduction with rerandomisation and sample splitting (CR-RSS). A set of indices $\mathcal{I}\subset\{1,\ldots,n\}$ of size $|\mathcal{I}|=29$ was drawn at random and fixed across Monte Carlo replications. This was used for model assessment for all procedures. The observations indexed by $\mathcal{I}^c$ were used for reduction for MS and CR-R, while CR-RSS used $\mathcal{I}$ for its first reduction and $\mathcal{I}^c$ for its second. For each approach, we recorded the proportion of times each signal variable was included in $\hat{\mathcal{S}}$, the proportion of times $\hat{\mathcal{S}}$ included the full set of signal variables, the proportion of times the set of signal variables was included in the model confidence set  $\mathcal{M}$ and the size of $\mathcal{M}$.

For a fair comparison, the comprehensive model was taken to consist of 15 variables for all procedures. Thus in the rerandomised versions of Cox reduction we retained in $\hat{\mathcal{S}}$ the 15 variables that were suggested most often, where the tuning parameter in the second reduction was chosen such that stability of the model confidence set was achieved. 
It is possible to construct examples in which a large significance level is needed in the second round in order that the same variables are retained frequently over successive rerandomisation. This is usually indicative of a weak signal to noise ratio and signals caution. Alternatively, there may be a sharp transition in the significance level at which the number of retained variables in all randomisations jumps sharply. Such instability is again a warning of partial fragility.

A favourable situation for both marginal screening and Cox reduction is when the estimated absolute effect appears larger than it is due to strong dependence between signal variables of an appropriate sign. This situation was explored and the results were similar to those reported in the analysis in Section S2 
in \cite{lewisSupp} and are therefore omitted. The aim of the following three examples is to highlight limitations of the methodology in the presence of certain covariate dependencies.

\subsubsection*{Example: effect cancellation}

A signal variable that is marginally uncorrelated in sample with the response variable is unlikely to be retained by marginal screening or Cox reduction. Section \ref{secMS} 
shows how this can happen for signal variables. In particular, it entails fairly strong sample correlation between signal variables of an appropriate sign. In this example, the set of signal variables was chosen to be the variables $\{1852, 3862, 4088\}$ with strong positive and negative correlations given by
\[\text{Corr}(X_{\mathcal{S}}, X_{\mathcal{S}})=\begin{pmatrix}
1.00 & -0.73 & -0.72\\
-0.73 & 1.00 & 0.73\\
-0.72 & 0.73 & 1.00
\end{pmatrix}.\]
Table \ref{real2} displays the results. The situation is not improved by increasing the signal to noise ratio, with both marginal screening and Cox reduction rarely including the first signal variable in the comprehensive model, as expected. If the effect cancellation was only partial and in a suitable range, Cox reduction will sometimes dominate over marginal screening, as discussed in Section \ref{posReduction}. 

\begin{center}
	\scalebox{0.9}{\begin{tabular}{ |c|c  c c| }
\hline
&\multirow{2}{*}{MS}&\multirow{2}{*}{CR-R}&\multirow{2}{*}{CR-RSS}\\
&&&\\
\hline
 $\mathbb{P}(1852 \in \hat{\mathcal{S}})$& 0.00 (0.00) & 0.00 (0.00) &0.00 (0.00) \\
$\mathbb{P}(3862 \in \hat{\mathcal{S}})$& 1.00 (0.00)  & 1.00 (0.06) & 1.00 (0.00) \\
$\mathbb{P}(4088 \in \hat{\mathcal{S}})$&1.00 (0.00) & 1.00 (0.00) & 0.99 (0.08) \\
\hline
 $\mathbb{P}(\mathcal{S} \subseteq \hat{\mathcal{S}})$& 0.00 (0.00) &0.00 (0.00)& 0.00 (0.00) \\
$|\hat{\mathcal{S}}|$&15 &15   & 15 \\
\hline
$\mathbb{P}(\mathcal{S} \in \mathcal{M})$& 0.00 (0.00) & 0.00 (0.00) & 0.00 (0.00) \\
$\mathbb{E}(|\mathcal{M}\backslash \mathcal{S}|)$&4280 (635)  & 2907 (1437)  & 3273 (1429)\\
\hline
\end{tabular} }
\captionof{table}{Simulation results using real data as the design matrix. The set of signal variables is $\mathcal{S}=\{1852,3862,4088\}$, $\hat{\mathcal{S}}$ is the comprehensive model, $\mathcal{M}$ is the set of models of size at most $s^\#=5$ that are significant at level $\alpha=0.01$. Empirical standard errors are given in brackets. }\label{real2}
\end{center}

\subsubsection*{Example: noise variables with stronger marginal correlations than signal}

A second example illustrates a situation in which Cox reduction outperforms marginal screening. The set of signal variables was taken as $\{5, 1812, 1861\}$ with correlation structure
\[\text{Corr}(X_{\mathcal{S}}, X_{\mathcal{S}})=\begin{pmatrix}
1.00 & 0.18 & 0.32\\
0.18 & 1.00 & 0.28\\
0.32 & 0.28 & 1.00
\end{pmatrix}.\]
Several noise variables were correlated with all three signal variables so that these appeared marginally stronger than the signal variables themselves. The signal variables nevertheless had  sufficient marginal strength to be retained through Cox reduction. Results are in Table \ref{real3}.

\begin{center}
	\scalebox{0.9}{	\begin{tabular}{ |c|c  c c| }
		\hline
		&\multirow{2}{*}{MS}&\multirow{2}{*}{CR-R}&\multirow{2}{*}{CR-RSS}\\
		&&&\\
		\hline
		$\mathbb{P}(5 \in \hat{\mathcal{S}})$& 0.00 (0.00) & 0.53 (0.50) &0.95 (0.21) \\
		$\mathbb{P}(1812 \in \hat{\mathcal{S}})$& 0.00 (0.06)  & 0.91 (0.28) & 1.00 (0.00) \\
		$\mathbb{P}(1861 \in \hat{\mathcal{S}})$&0.19 (0.39) & 0.65 (0.48) & 1.00 (0.05) \\
		\hline
		$\mathbb{P}(\mathcal{S} \subseteq \hat{\mathcal{S}})$& 0.00 (0.00) &0.28 (0.45)& 0.95 (0.22) \\
		$|\hat{\mathcal{S}}|$&15 &15   & 15 \\
		\hline
		$\mathbb{P}(\mathcal{S} \in \mathcal{M})$& 0.00 (0.00) & 0.27 (0.45) & 0.93 (0.23) \\
		$\mathbb{E}(|\mathcal{M}\backslash \mathcal{S}|)$&4854 (484)  & 2065 (1886)  & 555 (739)\\
		\hline
	\end{tabular} }
	\captionof{table}{As Table \ref{real2} with $\mathcal{S}$ replaced by $\{5,1812,1861\}$}\label{real3}
\end{center}

\subsubsection*{Example: weakly correlated signal variables}

This final example illustrates a setting where marginal screening outperforms Cox reduction. The set of signal variables was chosen to be a group of three weakly correlated variables indexed by $\{10,2027, 2923\}$. The sample correlations among these signal variables were
\[\text{Corr}(X_{\mathcal{S}}, X_{\mathcal{S}})=\begin{pmatrix}
1.00 & 0.10 & 0.08\\
0.10 & 1.00 & -0.04\\
0.08 & -0.04 & 1.00
\end{pmatrix}.\]
The results are reported in Table \ref{real4}.

\begin{center}
	\scalebox{0.9}{\begin{tabular}{ |c|c c c| }
\hline
&\multirow{2}{*}{MS}&\multirow{2}{*}{CR-R}&\multirow{2}{*}{CR-RSS}\\
&&&\\
\hline
 $\mathbb{P}(10 \in \hat{\mathcal{S}})$& 0.95 (0.21) & 0.34 (0.47) & 1.00 (0.00) \\
$\mathbb{P}(2027\in \hat{\mathcal{S}})$& 1.00 (0.00)  & 0.03 (0.16)  & 1.00 (0.21) \\
$\mathbb{P}(2923 \in \hat{\mathcal{S}})$&1.00 (0.00) & 0.72 (0.45) & 0.04 (0.10) \\
\hline
 $\mathbb{P}(\mathcal{S} \subseteq \hat{\mathcal{S}})$& 0.95 (0.21) & 0.00 (0.05) & 0.04 (0.20) \\
$|\hat{\mathcal{S}}|$&15& 15    &  15 \\
\hline
$\mathbb{P}(\mathcal{S} \in \mathcal{M})$& 0.94 (0.23) & 0.00 (0.05) & 0.04 (0.20) \\
$\mathbb{E}(|\mathcal{M}\backslash \mathcal{S}|)$&525 (935)  & 3653 (1615) & 732 (647)\\
\hline
\end{tabular} }
\captionof{table}{As Table \ref{real2} with $\mathcal{S}$ replaced by $\{10,2027,2923\}$}\label{real4}
\end{center}

Marginal screening successfully included all signal variables in the comprehensive model and produced a model confidence set of moderate size on average over Monte Carlo replicates. Cox reduction performed less well, with at least one signal variable being frequently omitted from $\hat{\mathcal{S}}$. In these rerandomised versions of Cox reduction, sample splitting increased the probability that the first two signal variables were retained but decreased the survival probability of the third. 

In all three examples, the omission of signal variables at the reduction stage does not affect the usefulness of the model confidence set for identifying models that fit the data equally well. 

\section{Further discussion}\label{discussion}
\subsection{Bootstrap confidence sets of models}\label{lassoRefute}
A natural attempt to construct confidence sets of models uses the LASSO or similar on $B$ bootstrap samples, obtained by sampling $n$ observations with replacement. This suggestion can be refuted in the light of \cite{SuBogdanCandes2017,Su2018}.

Let the true positive rate be the proportion of signal variables selected by the LASSO and the false discovery rate be the proportion of selected variables that are noise variables. \cite{SuBogdanCandes2017} showed that even under idealised conditions with large signal strengths and uncorrelated variables, the false discovery rate is lower bounded by a function of the true positive rate with probability one. Further, \cite{Su2018} shows that the first noise variable is selected earlier on the LASSO path than the final signal variable. Hence, each model either includes the full set of signal variables contaminated by noise variables, or does not include the full set of signal variables. These conclusions hold with probability tending to one in the notional double asymptotic regime in which $p$ grows with $n$.
As a result, with probability close to one, the true model is never selected and a model confidence set constructed in this way has coverage probability close to zero.
\subsection{Cox reduction with unknown variance}\label{unknownVar}
An unknown error variance $\sigma$ can either be estimated prior to Cox reduction based on all $p$ covariates or estimated anew in each of the regressions that are run. The former approach results in an unbiased estimate, whilst the latter typically produces biased estimates from regressions omitting signal variables. Nevertheless, the latter is preferable as a severely biased estimate indicates that the joint explanatory power of covariates appearing in a regression together is weak, making it less likely that these covariates are retained through the second round of Cox reduction. The associated Wald statistic is
\begin{eqnarray*}
	T_{Y:K}=\hat{\sigma}_{\mathit{K}}^{-1} D_{\mathit{K}}^{-1/2}\hat{\theta}_{Y:\mathit{K}}, \quad \quad
	\hat{\sigma}^2_{\mathit{K}}=\frac{\|Y-X_{\mathit{K}} \hat{\theta}_{Y:\mathit{K}}\|_2^2}{n-|\mathit{K}|}
\end{eqnarray*}
where $D_{\mathit{K}}$ is a diagonal matrix with entries given by the diagonal entries of $(X_{\mathit{K}}^TX_{\mathit{K}})^{-1}$. This choice does not affect the first-round reduction, which is based on the relative significance of variables in a given regression. Its effect in the second round is to de-emphasise covariates appearing in regressions with a weak signal.

If it is instead decided to estimate $\sigma$ prior to Cox reduction, suitable estimators are in \cite{FGH2012}, \cite{stadler2010} and \cite{dicker2014}. See \cite{reid2016} for a numerical comparison of their performance.  As the relative significance of covariates are unaffected by estimation of $\sigma$, an equivalent approach sets $\sigma$ to be an arbitrary value, adjusting the second-round significance level to give a stable comprehensive model of manageable size.

\subsection{Systematic arrangement of variables}\label{systematic}
Expression (\ref{eq:tStat_NS}) shows that differentiation between Cox reduction and marginal screening can be achieved by forcing second-round regressions to contain correlated signal and noise variables from $\mathcal{A}$. This suggests a more systematic arrangement of variables in the square in which indices $\mathit{E} $ and $\mathit{F}$ for which $|R(x_\mathit{E}, x_\mathit{F})|\gg0$ appear together in rows or columns. Further analysis is needed to ascertain any optimal arrangement of variables in the second round to maximise the probability of retaining signal variables.
\subsection{Binary responses}\label{secBinary}

Scheff{\'e} \cite{scheffe1959} emphasised the approximate validity of inference based on a notional linear model, owing to the randomisation. Lewis and Battey \cite{lewis2022} formalise this intuition beyond the context of Cox reduction when the outcomes are generated from a linear logistic model. Generative binary models, when fitted by maximum likelihood, are problematic in the second stage of Cox reduction, where many combinations of strong variables form a separating hyperplane, that is, produce no classification errors within sample. Ordinary least squares fitting overcomes the difficulties.

\subsection{Prediction}\label{prediction}

The argument for confidence sets of models is somewhat weakened when prediction is the primary goal. However, stability of the predictor over time and in different contexts is important and more likely with a model that has an underlying interpretation as well as immediate predictive success. By consideration of prediction intervals for each model in $\mathcal{M}$, we account both for model uncertainty and statistical uncertainty. Overlapping prediction intervals are reassuring, while those differing considerably point to instability of the predictor under alternative circumstances. See Section S3 of the supplementary material \cite{lewisSupp} for an example.

\bigskip

\noindent \textbf{Acknowledgements.} The work was partially supported by a UK Engineering and Physical Sciences Research Fellowship (to H.S.B). We are grateful to the anonymous referees and Associate Editors for their valuable suggestions.

\medskip

\noindent \textbf{Source code.} Source code for implementing Cox reduction and constructing confidence sets of models is available from \href{www.ma.imperial.ac.uk/~hbattey/softwareCube.html}{www.ma.imperial.ac.uk/$\sim$\hspace{-0.1em}hbattey/softwareCube.html}. The more user-interfaced R package \texttt{HCmodelSets} \cite{hhh} is accompanied by a detailed guide to usage \cite{Hoeltgebaum}. Source code implementing some of the refinements discussed in the present work, such as those used in the simulations in Section S2.1 of the supplementary material \cite{lewisSupp}, is available from \href{https://github.com/rm-lewis}{https://github.com/rm-lewis}.

\medskip 
\noindent \textbf{Supplementary Material. }The supplement provides pseudo-code for the procedure, simulation results and proofs of key results.




\end{document}